\documentclass[a4paper,UKenglish,cleveref, autoref, thm-restate]{lipics-v2021}

\newtheorem{reduction}{Reduction-Rule} 

\usepackage{tcolorbox}

\newcommand{\good}[0]{good}

\title{ A Refined Kernel for $d$-Hitting Set
} 

\titlerunning{A refined kernel for $d$-Hitting Set} 

\author{Yuxi Liu}{University of Electronic Science and Technology of China, Chengdu, China}{202211081321@std.uestc.edu.cn}{}{}

\author{Mingyu Xiao}{University of Electronic Science and Technology of China, Chengdu, China}{myxiao@uestc.edu.cn}{https://orcid.org/0000-0002-1012-2373}{}

\authorrunning{J. Open Access and J.\,R. Public}

\Copyright{J. Open Access and J.\,R. Public} 



\category{} 

\relatedversion{} 

\nolinenumbers 

\EventEditors{John Q. Open and Joan R. Access}
\EventNoEds{2}
\EventLongTitle{42nd Conference on Very Important Topics (CVIT 2016)}
\EventShortTitle{CVIT 2016}
\EventAcronym{CVIT}
\EventYear{2016}
\EventDate{December 24--27, 2016}
\EventLocation{Little Whinging, United Kingdom}
\EventLogo{}
\SeriesVolume{42}
\ArticleNo{23}

\begin{document}

\maketitle

\begin{abstract}
The \textsc{$d$-Hitting Set} problem is a fundamental problem in parameterized complexity, which asks whether a given hypergraph contains a vertex subset $S$ of size at most $k$ 
that intersects every hyperedge (i.e., $S \cap e \neq \emptyset$ for each hyperedge $e$).
The best known kernel for this problem, established by Abu-Khzam~\cite{DBLP:journals/jcss/Abu-Khzam10}, 
has $(2d - 1)k^{d - 1} + k$ vertices. This result has been very widely used in the literature as many problems can be modeled as a special $d$-Hitting Set problem. 
In this work, we present a refinement to this result by employing linear programming 
techniques to construct crown decompositions in hypergraphs. 
This approach yields a slight but notable improvement, reducing the size to $(2d - 2)k^{d - 1} + k$ vertices.
\end{abstract}

\section{Introduction}

The \textsc{$d$-Hitting Set} problem is a fundamental problem in theoretical computer science that has been extensively studied across multiple algorithmic paradigms for several decades. We begin with its formal definition:

\noindent\rule{\linewidth}{0.2mm}
\textsc{$d$-Hitting Set}\\
\textbf{Input:} A hypergraph $H=(V, E)$ where $|e|\leq d$ for all $e\in E$, and an integer $k$.\\
\textbf{Question:} Does there exist a vertex subset $S\subseteq V$ with $|S|\leq k$ such that $S\cap e \neq \emptyset$ for every $e\in E$?\\
\rule{\linewidth}{0.2mm}

This problem serves as a unifying framework for numerous graph modification problems. For any fixed graph $H$, the $H$-Free Vertex Deletion problem (deleting $k$ vertices to eliminate all $H$-subgraphs) is a special case of \textsc{$|V(H)|$-Hitting Set}. Important special cases include:
\begin{itemize}
    \item \textsc{Vertex Cover} (\textsc{2-Hitting Set}) when $H$ is a single edge.
    \item \textsc{Triangle-Free Vertex Deletion} and \textsc{Cluster Vertex Deletion} (\textsc{3-Hitting Set}) when $H=K_3$. 
    \item \textsc{$d$-Path Vertex Cover}, \textsc{$d$-Component Order Connectivity}, and \textsc{$d$-Bounded-Degree Vertex Deletion}.
\end{itemize}
The problem also has connections to descriptive complexity, where fragments of first-order logic can be reduced to \textsc{$d$-Hitting Set}~\cite{chen2017slicewise}.

Given its broad applicability, \textsc{$d$-Hitting Set} has been thoroughly investigated from multiple algorithmic perspectives. Our work focuses on kernelization results. The classical result, first established in 2004~\cite{fellows2008faster} and alternatively provable via the Erdős-Rado Sunflower Lemma~\cite{DBLP:series/txtcs/FlumG06}, yields a kernel with $O(k^d)$ vertices and hyperedges. This bound is essentially tight due to the following lower bound: 

\begin{proposition}[\cite{dell2014satisfiability}]\label{prop:lower-bound}
    For any constant integer $d > 1$, unless \textsf{co-NP} $\subseteq$ \textsf{NP}/poly, \textsc{$d$-Hitting Set} admits no polynomial-time reduction to instances of size $O(n^{d-\epsilon})$ for any $\epsilon > 0$.
\end{proposition}

While the $O(k^d)$ bound for hyperedges is optimal, the vertex kernel bound remained improvable. Abu-Khzam~\cite{DBLP:journals/jcss/Abu-Khzam10} achieved a breakthrough in 2007 with a vertex kernel of $(2d-1)k^{d-1} + k$, which has become a cornerstone result in parameterized algorithms due to the problem's wide applicability. A major open question in kernelization is whether \textsc{$d$-Hitting Set} admits an $O(k^{d-1-\epsilon})$ vertex kernel for some $\epsilon > 0$~\cite{bessy2011kernels,dom2010fixed,fomin2019kernelization,you2017approximate,fomin2023lossy}.

Recent work on lossy kernelization by Fomin et al.~\cite{fomin2023lossy} demonstrates that for any $\epsilon > 0$, \textsc{$d$-Hitting Set} admits a (randomized) pure $(d-\delta)$-approximate kernelization protocol producing kernels of size $O(k^{1+\epsilon})$, where the number of rounds and $\delta > 0$ are constants depending only on $d$ and $\epsilon$.

\vspace{2mm}
\noindent\textbf{The Contributions.}
In this paper, we show that \textsc{$d$-Hitting Set} admits a kernel with at most $(2d - 2)k^{d - 1} + k$ vertices. In particular, for $d = 3$, the \textsc{3-Hitting Set} problem admits a kernel with $4k^2 + k$ vertices. Our result improves upon the previous bound of $(2d - 1)k^{d - 1} + k$.

In the previous algorithm, 
one of the most important steps is to use a technique called \textit{HS crown decomposition} to reduce the instance size. 
In contrast, our algorithm does not rely on analyzing the combinatorial structure to identify HS crown decompositions. Instead, we employ a linear programming approach to compute them. This method enables us to find HS crown decompositions in smaller hypergraphs, ultimately leading to a more efficient kernel.

\vspace{2mm}
\noindent\textbf{The Organization.} 
Section~2 begins with a review of the fundamental definitions and the establishment of notation. In Section~3, we introduce the concepts of VC crown decomposition and HS crown decomposition, which serve as important tools in our algorithm. Section~4 presents and analyzes the algorithm, with the proof of the final reduction rule deferred to Section~5. Finally, Section~6 concludes the paper.

\section{Preliminaries}

In this paper, we denote a simple and undirected graph as $G$ and a hypergraph as $H$.
The input hitting set instance $(V, E)$ is treated as a hypergraph, where the vertices and hyperedges correspond to the elements of $V$ and $E$, respectively.
We use $V(H')$ and $E(H')$ to denote the vertex set and hyperedge set of a hypergraph $H'$, respectively.
For a graph $G$,
a vertex $v$ is called a \emph{neighbor} of a vertex $u$ if $\{u, v\} \in E(G)$.
Let $N(v)$ denote the set of neighbors of $v$.
For a vertex subset $X$, let $N(X) = \cup_{v\in X} N(v)\setminus X$ and $N[X] = N(X)\cup X$.

For a hypergraph $H$,
an \emph{independent set} is a set of vertices where no two vertices belong to the same hyperedge in $E(H)$.
A vertex subset $S$ \emph{hits} a hyperedge subset $E'\subseteq E(H)$ if $e \cap S \neq \emptyset$ holds for every $e\in E'$.
For each vertex $v$, we use $E(v)$ to denote the set of all hyperedges containing $v$.
For a hyperedge subset $E' \subseteq E(H)$, we use $V(E')$ to denote the set of vertices contained in the hyperedges of $E'$.
A \emph{subedge} of $H$ is a non-empty subset of $V$ that is contained in a hyperedge of $E(H)$.
An $i$-subedge is a subedge with exactly $i$ vertices.
For a subedge $e$, we use $E(e)$ to denote the set of all hyperedges containing $e$.
For a subedge set $E'$, we also use $V(E')$ to denote the set of vertices contained in the subedges of $E'$.

An instance of \textsc{$d$-Hitting Set} with input $H$ and $k$ is denoted as a pair $(H, k)$. For two instances $(H', k')$ and $(H'', k'')$, we say they are \emph{equivalent} if $(H, k)$ is a \textbf{yes}-instance if and only if $(H', k')$ is a \textbf{yes}-instance.
We assume that readers are familiar with kernelization. For basic concepts, please refer to references~\cite{cygan2015parameterized,fomin2019kernelization}.

\section{HS crown decomposition}



The classic crown decomposition method, known as VC crown decomposition, was originally introduced for the \textsc{Vertex Cover} problem \cite{DBLP:conf/alenex/Abu-KhzamCFLSS04,chor2004linear}.
In this paper, we will use a variant of it in hypergraphs.
We first provide the definition of VC crown decomposition for ease of reference.

\begin{definition}\label{VC-decomposition}
    A \textit{VC crown decomposition} of a graph $G = (V, E)$ is a pair $(I, J)$ with $I\subseteq V$ and $J\subseteq V$ satisfying the following properties.
    \begin{enumerate}
        \item $I$ is an independent set.
        \item There is no edge between $I$ and $V\setminus(I\cup J)$.
        \item There is an injective mapping (matching) $M: J\rightarrow I$ such that $\{x,M(x)\}\in E$ holds for all $x\in J$.
    \end{enumerate}
\end{definition}



The following lemma is used to find a VC crown decomposition in polynomial time if it exists.

\begin{lemma}[Lemma 9 in \cite{DBLP:conf/iwpec/KumarL16}]\label{q-expansion-finding}
    There exists a polynomial-time algorithm that given a bipartite graph $G :=((A, B), E)$,
    outputs (if it exists) two sets $\emptyset \neq I \subseteq A$ and $J\subseteq B$ with an injective mapping (matching) $M: J\rightarrow I$ such that $(I, J)$ is a VC crown decomposition for $G$.
\end{lemma}

To apply crown decomposition to \textsc{$d$-Hitting Set}, 
Abu-Khzam~\cite{DBLP:journals/jcss/Abu-Khzam10} introduced a variant of the crown decomposition called HS crown decomposition for hypergraphs,
which is defined below and illustrated in 
 Figure \ref*{dHS}.

\begin{figure}[!t]
    \centering
    \includegraphics[scale=0.5]{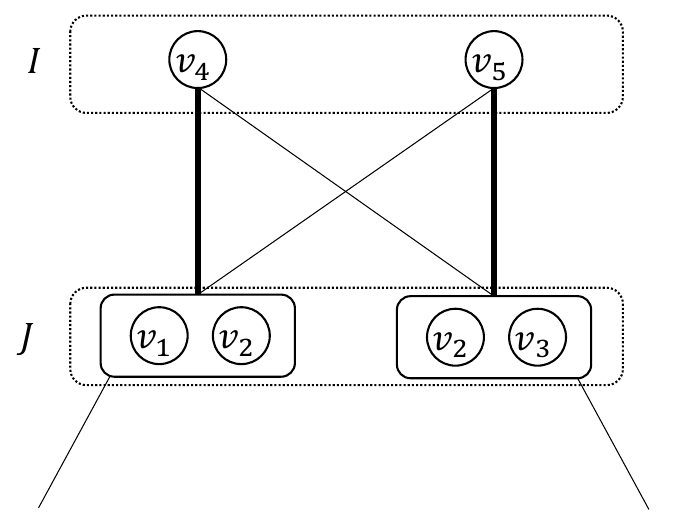}
    \caption{An example for HS crown decomposition $(I, J)$, where $d = 3$. The hypergraph contains hyperedges $\{v_1,v_2,v_4\}$, $\{v_1,v_2,v_5\}$, $\{v_2,v_3,v_4\}$ and $\{v_2,v_3,v_5\}$. $I = \{v_4, v_5\}$ and $J = \{\{v_1, v_2\}, \{v_2, v_3\}\}$. Thick edges correspond to the mapping $M$.}
    \label{dHS}
\end{figure}

\begin{definition}\label{HS crown decomposition}
    An \emph{HS crown decomposition} of a hypergraph $H = (V, E)$ is a pair $(I, J)$ such that:
    \begin{enumerate}
       \item $I\subseteq V$ is an independent set.
        \item $J=\{Y\subseteq V : Y\cup \{x\} \in E \text{ for some } x\in I \}$ is a subedge set.
        
        \item There is an injective mapping (matching) $M$ from $J$ to $I$ in the bipartite graph $(I \cup J, E')$ where $E' = \{(x, Y): x \in I, Y \in J \text{ and } (\{x\} \cup Y) \in E\}$.
    \end{enumerate}
\end{definition}

By the definition of HS crown decomposition $(I, J)$, there is no edge containing one vertex in $I$ and one vertex in $V \setminus (I \cup V(J))$. 
The following lemma shows that we can use HS crown decompositions to reduce the instance.

\begin{lemma}[\cite{DBLP:journals/jcss/Abu-Khzam10}]\label{$d$HS-correctness}
    Let $(H, k)$ be an instance of \textsc{$d$-Hitting Set} and $(I, J)$ be an HS crown decomposition of $H$.
    The instance $(H, k)$ is equivalent to the instance $(H':=(V(H) - I, (E(H) - E(I)) \cup J), k)$.
\end{lemma}

One of the most important steps is to effectively find HS crown decompositions and then we can use Lemma~\ref{$d$HS-correctness} to reduce the instance. If we can find an HS crown decomposition in polynomial time in a smaller hypergraph, this may lead to a smaller kernel. 
To do this, we use LP to help us find HS crown decompositions. Our LP optimal solution highlights the vertices in $I$ with value $0$ and the vertices in $V(J)$ with value $1$. Thus, we can find the HS crown decomposition $(I,J)$.
To do this, we also need to consider a more strict structure,
called the \textit{strict HS crown decomposition}.
An HS crown decomposition $(I, J)$ is \emph{strict} if $|I|\geq |J| + 1$.
Clearly, a strict HS crown decomposition satisfies the definition of an HS crown decomposition.
The following lemma shows how to find a strict HS crown decomposition if an independent set $I$ is given.
\begin{lemma}\label{strict-LOS-Lemma}
    Consider a hypergraph $H = (V, E)$ where the size of each hyperedge in $E$ is at least 2.
    Let $I \subseteq V$ be a non-empty independent set and $J:=\{Y\subseteq V : Y\cup \{x\} \in E \text{ for some } x\in I \}$.
    If $|I| \geq |J| + 1$, then there is a strict HS crown decomposition $(I', J')$ with $\emptyset \neq I' \subseteq I$
    and $J' \subseteq J$ together with an injective mapping (matching) $M$ from $J'$ to $I'$.
    Furthermore, this strict HS crown decomposition can be found in polynomial time if $I$ is given.
\end{lemma}

\begin{proof}
    We construct an auxiliary bipartite graph $G' = (V', E')$ with $V' = (A', B')$ as follows:
    Each vertex in $A'$ corresponds to a vertex in the vertex set $I$, and each vertex in $B'$ corresponds to a subedge in the subedge set $J$.
    A vertex $a \in A'$ is adjacent to a vertex $b \in B'$ if and only if the union of the vertex $a$ in $H$ and the subedge corresponding to $b$ is a hyperedge in $H$.
    Since $|I| \geq |J| + 1$, we have that $|A'|\geq |B'| + 1$.
    Note that $|J| > 0$ since the size of each hyperedge in $E$ is at least 2.
    We compute a maximum matching $M' \subseteq E'$ of $G'$ by using the $O(n^{5/2})$-time matching algorithm~\cite{hopcroft1973n}.

    A vertex is called $M'$-saturated if it is an endpoint of an edge in $M'$.
    A path in $G'$ that alternates between edges not in $M'$ and edges in $M'$ is called an $M'$-alternating path.
    Note that a matching in $G'$ corresponds to a mapping $M$ from $J$ to $I$ in the original hypergraph $H$.
    We use $Z \subseteq A'$ to denote the set of vertices in $A'$ which are not $M'$-saturated.
    Since $|A'| \geq |B'| + 1$, we have that $Z \neq \emptyset$.
    Let $I'' \subseteq A' $ be the set of vertices in $A'$ that are reachable from a vertex in $Z$ via an $M'$-alternating path, possibly of length zero (which means that $Z \subseteq I''$).
    Let $J'' \subseteq B'$ be the set of vertices in $B'$ that are reachable from a vertex in $Z$ via an $M'$-alternating path in $G'$.
    See Figure \ref*{strict-LOS-Lemma_pic} for an illustration.

    Let $I' \subseteq I$ be the vertex set in the hypergraph $H$ corresponding to $I''$, and let $J' \subseteq J$ be the subedge set in the hypergraph $H$ corresponding to $J''$.
    We will show that $(I', J')$ is a strict HS crown decomposition of $H$.

    \begin{figure}[!t]
        \centering
        \includegraphics[scale=0.6]{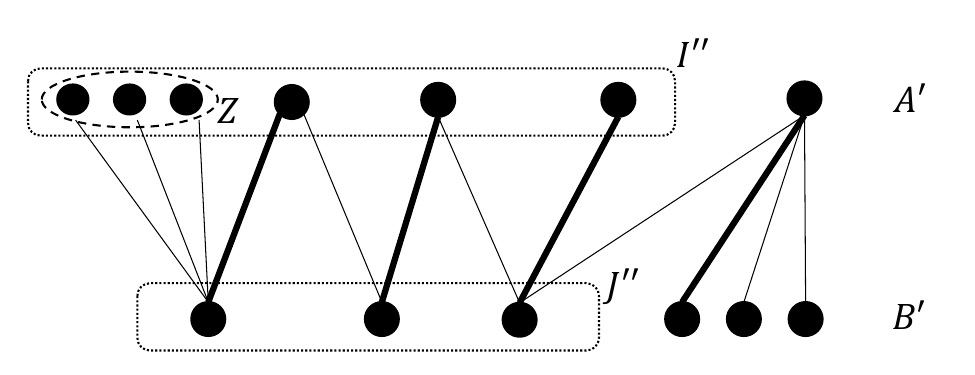}
        \caption{An illustration of $I''$ and $J''$, where thick edges denote the edges in $M'$.}
        \label{strict-LOS-Lemma_pic}
    \end{figure}

    The first condition of HS crown decomposition is satisfied since $I$ is an independent set and $I' \subseteq I$.
    Then, we consider the second condition.
    Consider an edge $(u, v)$ in $G'$ where $u \in I''$.
    If $(u, v) \in M'$, then $v$ must be in $J''$ since $u$ is $M'$-saturated and therefore the last edge of an $M'$-alternating path from a vertex in $Z$ to $u$ must be $(u, v)$.
    If $(u, v) \notin M'$, then $v$ is also in $J''$.
    The reason is that if an $M'$-alternating path $P$ begins from a vertex in $Z$ and ends at the vertex $u \in I''$, 
    then extending $P$ by adding the edge $(u, v)$ at the ending yields another $M'$-alternating path and then $v$ is also reachable from a vertex in $Z$ via an $M'$-alternating path.
    Therefore, we have that $N(I'') \subseteq J''$.
    By the definition of $J''$, each vertex in $J''$ is adjacent to some vertices in $I''$.
    So, we have that $N(I'') = J''$.
    Thus, Let $J_{I'} := \{Y \subseteq V : Y \cup \{x\} \in E \text{ for some }x \in I' \}$. We have that $J_{I'} = J'$, which satisfies the first condition of HS crown decomposition. 

    Next, we consider the third condition.
    We claim that all vertices in $J''$ are $M'$-saturated.
    The reason is that a matching $M'$ is maximum if and only if there is no $M'$-augmenting path (an $M'$-alternating path beginning and ending at non-$M'$-saturated vertices) by Berge's lemma~\cite{berge1957two}.
    Furthermore, we claim that any vertex $v \in A'$ matched to a vertex $u \in J''$ is in $I''$.
    The reason is that if an $M'$-alternating path $P$ begins from a vertex in $Z$ and ends at a vertex $u \in J''$, then extending $P$ by adding the edge $(u, v)$ results in another $M'$-alternating path.
    Thus, there exists a matching $M'' \subseteq M'$ of size $|J''| = |J'| > 0$ with $V (M'') \subseteq J'' \cup I''$. The injective mapping in $G'$ corresponding to $M''$ satisfies the third condition. 
    Since $I''$ contains $|J''|$ vertices matched to $J''$ and at least one vertex in $Z$
    (Since $Z \neq \emptyset$), we have that $I'\neq \emptyset$ and $|I'|\geq |J'| + 1$. The lemma holds. 
\end{proof}

\section{The Algorithm Framework}
Our kernelization algorithm contains six reduction rules. When we introduce one reduction rule, we assume that all previous reduction rules are no longer applicable to the current instance.
As mentioned before, our kernelization algorithm is a refinement of the previous algorithm in \cite{DBLP:journals/jcss/Abu-Khzam10}, we would like to show the difference between these two algorithms first. See Figure~\ref{framework} for an illustration.
Rules 1, 2, 4, and 5 are the same. Rule 3 is simple but newly added in this paper. The major difference and contribution are in Rule 6, where we use a new method to find HS crown decomposition. In fact, the first four rules are simple and basic rules.
We will use ${\tt Reduce}(H, k)$ to denote our kernelization algorithm. Next, we introduce the rules.

\begin{figure}[!t]
    \centering
    \includegraphics[height = 4.2cm]{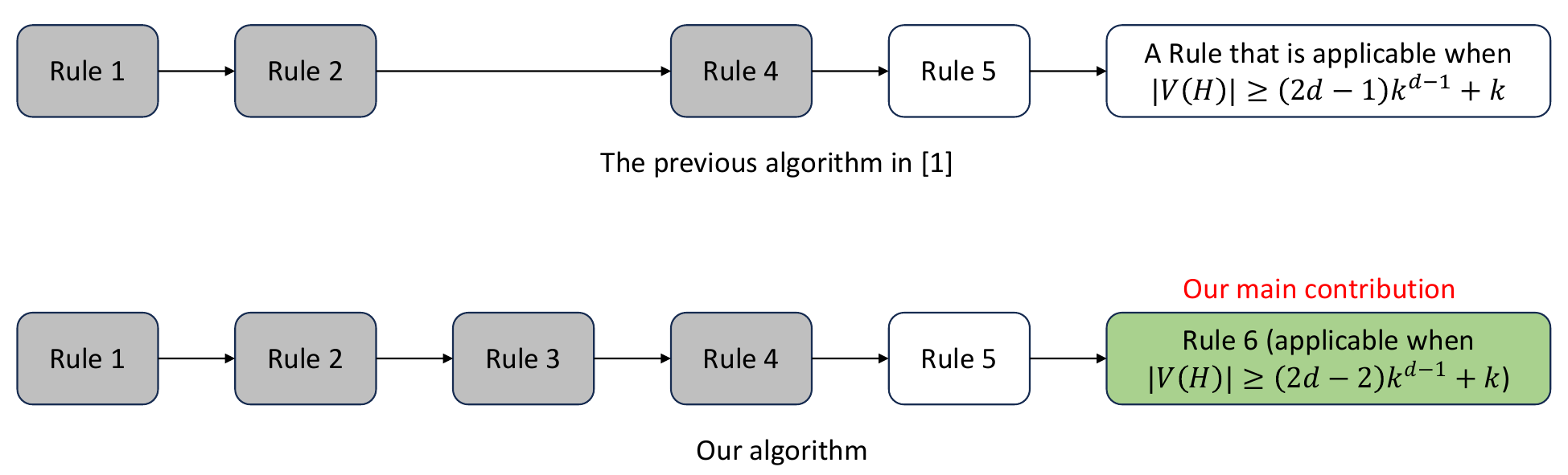}
    \caption{A comparison between the previous algorithm and our algorithm. The rules marked gray are simple and basic rules and Rule 6 marked green is our main contribution.}
    \label{framework}
\end{figure}

\begin{reduction}\label{basic-reduction-1}
    If there exists $x,y\in V$ such that $E(x)\subseteq E(y)$, return ${\tt Reduce}(H':=(V(H) - \{x\}, E(H) - E(x)), k)$.
\end{reduction}

\begin{reduction}\label{basic-reduction-2}
    If there exists $e_1,e_2\in E$ such that $V(e_1)\subseteq V(e_2)$, return ${\tt Reduce}(H':=(V(H), E(H) - \{e_2\}), k)$.
\end{reduction}

\begin{reduction}\label{basic-reduction-3}
    If there exists $e\in E$ such that $|V(e)| = 1$, return ${\tt Reduce}(H':=(V(H) - V(e), E(H) - \{e\}), k - 1)$.
\end{reduction}

The correctness of Reduction Rules 1–\ref{basic-reduction-3} is straightforward.

\begin{reduction}\label{highdegree-reduction}
    If there exists a subedge $e \subseteq V$ such that $|e| = d - 2$ and there are more than $k$ hyperedges whose pairwise intersection is $e$,
    then return ${\tt Reduce}(H':=(V(H), (E(H) - E(e)) \cup \{e\}), k)$.
\end{reduction}


For Reduction Rule \ref{highdegree-reduction}, the correctness follows from the observation that any minimum hitting set must include at least one vertex from the subedge $e$.



We call two hyperedges $e_1$ and $e_2$ are \emph{weakly related} if $|e_1\cap e_2|\leq d - 2$. 
For a subedge $e$, let $W(e)$ denote the set of hyperedges in $W$ that contain~$e$.
The following rule first identifies a maximal collection of weakly related edges $W$ and then reduces the instance based on 
$W$.

\begin{reduction}\label{highoccu-reduction}

    If the most recently applied reduction rule is not Reduction Rule \ref*{highoccu-reduction}, we first use a greedy algorithm to identify an arbitrary maximal set of weakly related edges $W$, and then proceed with the following algorithm:

    \begin{tcolorbox}  %
        \begin{enumerate}
            \item For i = $d - 2$ to 1 do
            \item ~~~~For each $i$-subedge $e$ of $W$:
            \item ~~~~~~~~If $|W(e)| > k^{d - 1 - i}$, then delete  all hyperedges containing $e$ from $E(H)$, add $e$ to $E(H)$. 
        \end{enumerate}

    \end{tcolorbox}
    Let $(H', k')$ be the instance  after running this algorithm (it is possible that $H' = H$ and $k' = k$). Return ${\tt Reduce}(H', k')$.
\end{reduction}

Reduction Rule \ref*{highoccu-reduction} was introduced in~\cite{DBLP:journals/jcss/Abu-Khzam10}. 
Our algorithm may apply Rule 5 for several times, while the algorithm in~\cite{DBLP:journals/jcss/Abu-Khzam10} only applies it for once. 
The following property holds as a direct consequence.
\begin{proposition}\label{size-Prop}
    Consider an instance $(H, k)$ where Reduction Rules 1-\ref*{highoccu-reduction} cannot be applied.
    There exists a maximal collection of weakly related edges $W$ such that any 1-subedge is contained in at most $k^{d - 2}$ hyperedges in $W$.
\end{proposition}

Now we are ready to present that last reduction rule. 

\begin{reduction}\label{lp-reduction}
    If $|V(H)|\geq (2d - 2)k^{d - 1} + k + 1$, then call the polynomial time algorithm in Lemma \ref*{RRlp-lemma} to compute an HS crown decomposition $(I, J)$.
    Return ${\tt Reduce}(H':=(V(H) - I, (E(H) - E(I)) \cup J), k)$.
    If the algorithm cannot find such an HS crown decomposition, return \textbf{no}.
\end{reduction}

We postpone the presentation of the polynomial-time algorithm for Reduction Rule \ref{lp-reduction} and the correctness of Reduction Rule \ref{lp-reduction} to the next section. For now, assuming the correctness of Reduction Rule \ref{lp-reduction}, we proceed to analyze the kernel size and running time.

In each reduction rule, our kernelization algorithm either calls itself recursively or concludes \textbf{no} to indicate that no solution exists.
When Reduction Rule \ref{lp-reduction} cannot be applied, we have that $|V(H)|\leq (2d - 2)k^{d - 1} + k$.
Thus, the kernel size is at most $ (2d - 2)k^{d - 1} + k$.

Next, we show that our kernelization algorithm runs in polynomial time.
It is easy to see that each execution of Reduction Rules 1-3 and \ref*{highoccu-reduction} takes polynomial time.
For Reduction Rule \ref*{highdegree-reduction}, consider a $(d - 2)$-subedge $e$, and let $E_e$ be the hyperedge set properly containing $e$. 
We construct an auxiliary simple graph $G' = (V', E')$ where $V':= V(E_e)\setminus e$ and $E'(e):=\{e': e\cup e'\in E(H)\}$.
The subedge $e$ is the pair-wise intersection of more than $k$ hyperedges if and only if there exists a matching of size greater than $k$. 
Since finding a maximum matching in a simple graph can be solved in polynomial time~\cite{edmonds1965paths}, 
Reduction Rule \ref*{highdegree-reduction} runs in polynomial time.
The execution of Reduction Rule \ref{lp-reduction} will run a polynomial time algorithm given in Lemma \ref{RRlp-lemma}.
Thus, each execution of all Reduction Rules takes polynomial time.

We further analyze the number of times each reduction rule is applied.
For Reduction Rules 1-\ref*{highdegree-reduction}, after a single execution, the number of hyperedges or the number of vertices decreases by at least one. 
Therefore, Reduction Rules 1-\ref*{highdegree-reduction} are applied at most $|V(H)| + |E(H)|$ times.
For Reduction Rule \ref*{lp-reduction}, after a single execution, either the number of vertices decreases by at least one or the instance is  concluded \textbf{no}.
Therefore, Reduction Rules \ref*{lp-reduction} will be applied at most $|V(H)|$ times.
For Reduction Rule \ref*{highoccu-reduction}, after a single execution, another reduction rule is applied once (if there is no applicable reduction rule, the algorithm exits).
Therefore, Reduction Rule~\ref*{highoccu-reduction} will be applied at most $|V(H)| + 2|E(H)|$ times.
Since the total number of applications of all reduction rules is polynomially bounded, our kernelization algorithm runs in polynomial time.
We have the following theorem.

\begin{theorem}
    \textsc{$d$-Hitting Set} admits a kernel with $(2d - 2)k^{d - 1} + k$ vertices.
\end{theorem}

\section{Reduction Rule \ref*{lp-reduction}}

In this section, we introduce the polynomial-time algorithm used in Reduction Rule \ref{lp-reduction}. For a given instance 
$(H, k)$, we address two subproblems:
(i) Under what conditions can we guarantee the existence of an HS crown decomposition in $H$?
(ii) If such a decomposition exists, under what conditions can we efficiently find it?
The first subproblem is addressed in Section 5.1, while the second is discussed in Section 5.2, where we also prove the correctness of Reduction Rule \ref{lp-reduction}.

\subsection{Existence of HS Crown Decompositions}
In fact, in this subsection, we will show the condition of the existence of strict HS crown decompositions. Once we know there exist strict HS crown decompositions, we will show how to compute HS crown decompositions in the next subsection.

\begin{definition}
    In a hypergraph $H = (V, E)$, a hyperedge subset $W \subseteq E$ is \emph{\good} if there exists a maximal collection of weakly related edges $W'$ such that $W' \subseteq W$.
\end{definition}

Recall that two hyperedges $e_1$ and $e_2$ are weakly related if $|e_1\cap e_2|\leq d - 2$.
Given a {\good} hyperedge subset $W$, we have that every hyperedge in $E$ is either included in $W$ or shares at least $d - 1$ vertices with an element of $W$.

Suppose $(H, k)$ is a \textbf{yes}-instance.
We analyze which conditions can guarantee the existence of an HS crown decomposition.
Since we are only concerned with the existence of an HS crown decomposition, we assume that a minimum solution $S$ is given.
First, we construct a {\good} hyperedge subset $W$ such that $S\subseteq V(W)$, $|W|\leq k^{d - 1}$ and each vertex in $S$ is contained in at most $k^{d - 2}$ hyperedges in $W$.
By the definition of {\good} hyperedge subsets, we have that vertex set $I = V(H)\setminus V(W)$ is an independent set.
Let $J:=\{Y\subseteq V : Y\cup \{x\} \in E \text{ for some } x\in I \}$.
By Lemma \ref*{strict-LOS-Lemma}, if $|I|\geq |J| + 1$, then there exists a strict HS crown decomposition $(I', J')$ where $\emptyset \neq I'\subseteq I$ and $J'\subseteq J$. 
Finally, in Lemma \ref*{strict-key-bound}, we use Lemma \ref*{maximal-weakly-subedge-lemma} to bound the size of $J$ and Lemma \ref*{strict-LOS-Lemma} to show that either an HS crown decomposition exists or the size of $I$ is bounded.

\begin{lemma}\label{maximal-weakly-subedge-lemma}
    Let $(H, k)$ be a \textbf{yes}-instance such that Reduction Rules 1-\ref*{highoccu-reduction} cannot be applied, and let $S$ be a minimum solution to $(H, k)$.
    There exists a {\good} hyperedge subset $W$ such that 
    \begin{enumerate}
        \item $S\subseteq V(W)$;
        \item $|W|\leq k^{d - 1}$;
        \item Each vertex in $S$ is contained in at most $k^{d - 2}$ hyperedges in $W$.
    \end{enumerate}

\end{lemma}

\begin{proof}
    After Reduction Rule \ref*{highoccu-reduction}, we have that there exists a maximal collection of weakly related edges $W'$ such that any 1-subedge is contained in at most $k^{d - 2}$ hyperedges in $W'$ by Proposition \ref*{size-Prop}.
    Since $(H, k)$ is a \textbf{yes}-instance, we have that $|S|\leq k$.
    We consider the following two cases.
    
    Case 1. $S\subseteq V(W')$: 
    By Proposition \ref*{size-Prop}, since vertices in $S$ can be treated as 1-subedges, we have that each vertex in $S$ is contained in at most $k^{d - 2}$ hyperedges in $W'$.
    Since $S$ hits all hyperedges in $W'$,
    we have that $|W'|\leq k^{d - 2}|S|\leq k^{d - 1}$.
    Clearly, every maximal collection of weakly related edges is \good.
    Thus, $W'$ is a {\good} hyperedge subset satisfying the three conditions.

    Case 2. $S\nsubseteq V(W')$: 
    Let $S_1 = V(W')\cap S$ and $S_2 = S\setminus S_1$.
    Since $S$ is a minimum solution, we have that for each vertex $v\in S_2$, there exists a hyperedge $e$ such that $e\cap S = \{v\}$.
    By the definition of $S_2$, we have that such $e$ is not in $W'$.
    Now, we construct a {\good} hyperedge subset $W$ by adding an arbitrary hyperedge $e$ containing $v$ with $e\cap S = \{v\}$ for each vertex $v$ in $S_2$.
    Clearly, it holds that $S\subseteq W$.
    Next, we consider the size of $W$.
    By similar arguments, each vertex in $S_1$ is contained in at most $k^{d - 2}$ hyperedges in $W'$.
    Since $S_1$ hits all hyperedges in $W'$ and $|S_1|<|S|\leq k$, we have that $|W'|\leq (k - 1)k^{d - 2}$.
    Thus, we have that $|W|\leq |W'| + |S_2| \leq k^{d - 1}$, satisfying the second condition.
    Finally, we further consider the last condition.
    For each vertex $v$ in $S_1$, since the hyperedges in $W\setminus W'$ only contain solution vertices from $S_2$ by our construction, 
    we have that $v$ is contained in at most $k^{d - 2}$ hyperedges in $W$.
    For each vertex $v$ in $S_2$, since $v\notin W'$, $v$ is contained in at most $|W\setminus W'| = |S_2| \leq k$ hyperedges in $W$.
    Thus, the last condition holds.
    Since each superset of a maximal collection of weakly related edges is also \good, it holds that $W$ is a {\good} hyperedge subset satisfying the three conditions.

    Thus, this lemma holds.
\end{proof}

By Lemma \ref*{maximal-weakly-subedge-lemma} and Lemma \ref*{strict-LOS-Lemma}, we have the following lemma to show that if the size of the vertex set of the hypergraph is big enough, there exists a strict HS crown decomposition.
This size condition is strongly related to our kernel size.

\begin{lemma}\label{strict-key-bound}
    Let $(H, k)$ be a \textbf{yes}-instance where Reduction Rules 1-\ref{highoccu-reduction} cannot be applied.
        If $|V(H)|\geq (2d - 2)k^{d - 1} + k + 1$, then there exists a strict HS crown decomposition in $H$.

\end{lemma}

\begin{proof}

    For any minimum solution $S$,
    Lemma \ref*{maximal-weakly-subedge-lemma}, guarantees the existence of a {\good} hyperedge subset $W$ such that $S\subseteq V(W)$, $|W|\leq k^{d - 1}$ and each vertex in $S$ is contained in at most $k^{d - 2}$ hyperedges in $W$.
    For any hyperedge $e\in E(H)$, we have that $|e \cap S|\geq 1$ since $S$ is a feasible solution.
    We construct a subedge set $J_0$ by the following method.

    \begin{tcolorbox}  
        \begin{enumerate}
            \item For each vertex $s$ in $S$:
            \item ~~~~For each hyperedge $e$ in $E(s)\cap W$:
            \item ~~~~~~~~If $|e| = d$:
            \item ~~~~~~~~~~~~Let $ e = \{s, v_1, v_2, \cdots, v_{d - 1}\}$.
            \item ~~~~~~~~~~~~For each $i$ from $1$ to $d - 1$:
            \item ~~~~~~~~~~~~~~~~Add $e\setminus \{v_i\}$ to $J_0$.
            \item ~~~~~~~~Else if $|e| = d - 1$:
            \item ~~~~~~~~~~~~Add $e$ to $J_0$.
        \end{enumerate}

    \end{tcolorbox}

    Since each vertex in $S$ is contained in at most $k^{d - 2}$ hyperedges in $W$, 
    we have that $|J_0|\leq (d - 1)k^{d - 1}$. 
    Let $I = V\setminus V(W)$, let $J:=\{Y\subseteq V : Y\cup \{x\} \in E \text{ for some } x\in I \}$. 
    To bound the size of $J$, we show that $J\subseteq J_0$.

    By the definition of $W$, there exists a maximal collection of weakly related edges $W'$ such that $W'\subseteq W$.
    Clearly, by the definition of $W'$, $I' = V\setminus V(W')$ is an independent set. 
    For each hyperedge $e$ containing a vertex $v_e$ in $I'$, we claim that $|e| = d$.
    The reason is that if $|e|\leq d - 1$, then $e$ shares at most $d - 2$ vertices with each element of $W'$ since $v_e$ is not contained in any hyperedge in $W'$.
    Thus, $W'\cup \{e\}$ is a collection of weakly related edges, which contradicts the maximality of $W'$.

    Now we consider the vertex set $I$.
    Since $I\subseteq I'$, it holds that $I = V\setminus V(W)$ is an independent set.
    Moreover, for each hyperedge $e$ containing one vertex $v_e$ in $I$, we have $|e| = d$.
    We further have that for each subedge $e'\in J$, $|e'| = d - 1$.
    Since $S$ is a feasible solution, there is no hyperedge containing one vertex in $I$ and other vertices not in $S$.
    So we have that $J\subseteq J_0$, which implies that $|J|\leq |J_0|$.

    Since $S$ is a feasible solution, for any hyperedge $e \in W$, we have that $|e\setminus S| \leq d - 1$.
    Therefore, we have that $|V(W)|\leq k + (d - 1)|W| \leq k + (d - 1)k^{d - 1}$.

    Suppose that $(H, k)$ is a \textbf{yes}-instance and $|V(H)|\geq (2d - 2)k^{d - 1} + k + 1$.
    Since $|V(W)|\leq (d - 1)k^{d - 1} + k$, we have that $|I|\geq (d - 1)k^{d - 1} + 1$.
    Since $|J|\leq |J_0|\leq (d - 1)k^{d - 1}$, we have that $|I|\geq |J| + 1$.
    By Lemma \ref*{strict-LOS-Lemma}, there exists a strict HS crown decomposition in $H$.
    This lemma holds.
\end{proof}

\subsection{Computing HS Crown Decompositions}
Next, we consider how to efficiently compute an HS crown decomposition, assuming there exists a strict HS crown decomposition.
We will use an LP method to solve this problem.

We encode the given hypergraph $H$ as an LP model (SHS-LP):
\[
    \begin{split}
        \mbox{Min } &\sum_{v\in V} x_v\\  
           \mbox{subject to } & \sum_{v'\in e} x_{v'}\geq |e|-1, \text{ for each } e\in E\\     
        & 0\leq x_v \leq 1,\quad \forall v\in V.
    \end{split}
\]

Note that SHS-LP is not an LP relaxation of the integer programming model for the \textsc{$d$-Hitting Set} problem.
For \textsc{$d$-Hitting Set}, we only need to require that 
\[\sum_{v'\in e} x_{v'}\geq1, \text{ for each } e\in E.\]
In our algorithm, we indeed need to use SHS-LP.

We use $S_L$ to denote an optimal solution to SHS-LP with the optimal value being $L$.
For a vertex set $X \subseteq V(H)$, we define $X = 1$ to indicate that each variable corresponding to the vertices in $X$ is set to $1$. 
For a subedge $e$, its value is defined as $\sum_{v\in e} x_v$.
For a hyperedge $e$, if there is a vertex $v$ in $e$ such that $x_v = 0$, then we have that for each $u\in e\setminus \{v\}$, $x_u = 1$.


We call a strict HS crown decomposition $(I, J)$ \emph{minimal} if there is no strict HS crown decomposition $(I', J')$ such that $I'\subset I$.
For any minimal strict HS crown decomposition, the following lemma shows that the LP model SHS-LP  will highlight the vertices in $I$ with value $0$ and the vertices in $V(J)$ with value $1$.

\begin{lemma}\label{LP-lemma}
    Suppose that $(H, k)$ is a \textbf{yes}-instance of \textsc{$d$-Hitting Set} where Reduction Rules 1-\ref*{highoccu-reduction} cannot be applied.
    Let $(I, J)$ be a minimal strict HS crown decomposition, and let $S_L$ be an optimal LP solution.
    $S_L$ sets all variables corresponding to vertices in $I$ to $0$ and sets all variables corresponding to vertices in $V(J)$ to $1$.
\end{lemma}

\begin{proof}
    For any optimum solution $S_L$, let $I_1\subseteq I$ be the set of all vertices in $I$ whose variables are set to $0$ by $S_L$, and let $I_2 = I\setminus I_1$.
    Let $J_1 := \{Y\subseteq V : Y\cup \{x\} \in E \text{ for some } x\in I_1 \}$, let $J_2 = J\setminus J_1$.
    For any hyperedge $\{v_1, v_2, \cdots, v_l\}$, the constraint $x_{v_1} + x_{v_2} + \cdots + x_{v_l} \geq (l - 1)$ implies that if $x_{v_1} = 0$, then for each $i = 2, 3, \cdots, l$, $x_{v_i} = 1$.
    Thus, by the definition of $I_1$, we have that for each vertex $v\in V(J_1)$, $x_v = 1$.
    We claim $I_2 = \emptyset$ to conclude the proof.
    By contradiction, suppose $I_2 \neq \emptyset$.

    If $|I_1|\geq |J_1| + 1$, by Lemma \ref*{strict-LOS-Lemma}, there exists a strict HS crown decomposition $(I', J')$ with $I' \subseteq I_1$ and $J' \subseteq J_1$ together with an injective mapping (matching) $M$ from $J'$ to $I'$, which contradicts the minimality of the strict HS crown decomposition $(I, J)$.

    If $|I_1|\leq |J_1|$, then $|I_2|\geq |J_2| + 1$.
    Let $\epsilon = \min\{x_v|v\in I_2\}$.
    Now we construct a new LP solution $S_{L'}$ by the following method.
    See Figure \ref*{RPcase2} for an illustration.

    \begin{tcolorbox}  
        \begin{enumerate}
            \item For each vertex $u$ in $I_2$:
            \item ~~~~$x_u\leftarrow x_u - \epsilon$.
            \item For each subedge $e$ in $J_2$:
            \item ~~~~Let $y_1$ be $\epsilon$.
            \item ~~~~~~~~For each vertex $u$ in $e$:
            \item ~~~~~~~~~~~~Let $y_2$ be $\min\{1 - x_u, y_1\}$.
            \item ~~~~~~~~~~~~$x_u\leftarrow x_u + y_1$.
            \item ~~~~~~~~~~~~$y_1\leftarrow y_1 - y_2$.
        \end{enumerate}

    \end{tcolorbox}

    \begin{figure}[!t]
        \centering
        \includegraphics[height=4.5cm]{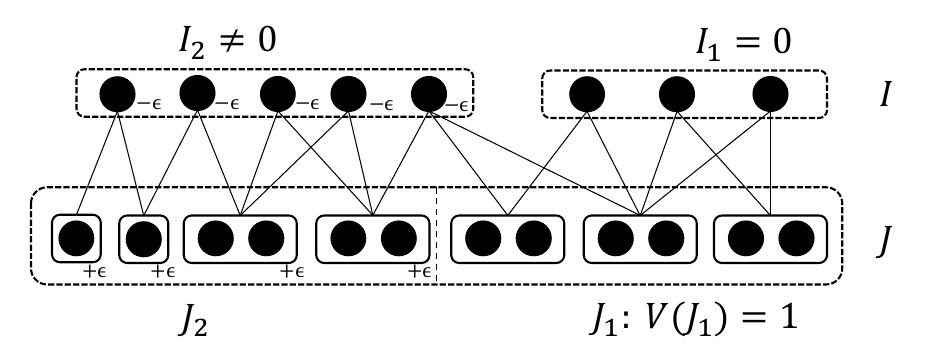}
        \caption{The case that $|I_1|\leq |J_1|$ in the proof of Lemma \ref*{LP-lemma}, where $d = 3$.}
        \label{RPcase2}
    \end{figure}

    This construction method for $S_{L'}$ ensures that for each subedge $e$ in $J_2$, either the value of $e$ increases by $\epsilon$ or the value of each vertex in $e$ is equal to 1.
    Now, we show that $S_{L'}$ is a feasible LP solution. 
    By the definitions of $I$ and $J$, we have that for every hyperedge $e$ containing a vertex $v$ in $I_2$, $e$ must contain one subedge $e\setminus \{v\}$ in $J$.
    Recall that the value of each vertex in $V(J_1)$ is equal to $1$.
    And for each subedge $e$ in $J_2$, either the value of $e$ increases by $\epsilon$ or the value of each vertex in $e$ is equal to 1.
    Thus, $S_{L'}$ is a feasible LP solution. 
    Since $|I_2|\geq |J_2| + 1$, we have that $L' < L$, which contradicts the fact that $S_L$ is an optimal LP solution.

    Thus, we have that $I = 0$.
    By the definition of our LP model, we have that $V(J) = 1$. This lemma holds.
\end{proof}


For any optimum solution $S_L$ to SHS-LP, let vertex sets $A := \{v\in V| x_v = 0\}$ and $V_B := \{v\in V| x_v = 1\}$.
Let a subedge set $B:= \{Y \subseteq V_B : Y \cup \{x\} \in E \text{ for some }x \in A \}$.
Recall that the constraints in SHS-LP ensure that for any hyperedge $\{v_1, v_2, \cdots, v_l\}$ in $E(H)$, there is at most one vertex whose value is set to $0$ by $S_L$.
Thus, we have that $A$ is an independent set.
If there exists a strict HS crown decomposition in $H$,
by Lemma \ref*{LP-lemma}, we can claim that there exists a (strict) HS crown decomposition $(I, J)$ such that $I\subseteq A, J\subseteq B$.
We present the following lemma.


\begin{lemma}\label{RRlp-lemma}
    There exists a polynomial-time algorithm that,
    given an instance $(H, k)$ of \textsc{$d$-Hitting Set} on at least $(2d - 2)k^{d - 1} + k + 1$ vertices where Reduction Rules 1-\ref*{highoccu-reduction} cannot be applied,
    either finds an HS crown decomposition $(I, J)$ with $I \neq \emptyset$ or concludes that $(H, k)$ is a \textbf{no}-instance.
\end{lemma}

\begin{proof}
    The algorithm works as follows:

    First, solve SHS-LP in polynomial time and get an optimal solution $S_L$.
    Let $A := \{v\in V| x_v = 0\}$ and $V_B := \{v\in V| x_v = 1\}$.
    Let $B:= \{Y \subseteq V_B : Y \cup \{x\} \in E \text{ for some }x \in A \}$.
    We construct an auxiliary bipartite graph $G' = (V', E')$ with $V' = (A', B')$ as follows:
    Each vertex in $A'$ corresponds to a vertex in the vertex set $A$, and each vertex in $B'$ corresponds to a subedge in the subedge set $B$.
    And a vertex $a \in A'$ is adjacent to a vertex $b \in B'$ if and only if the union of the vertex $a$ in $H$ and the subedge corresponding to $b$ is a hyperedge in $H$.

    Second, call the polynomial-time algorithm in Lemma \ref*{q-expansion-finding} to compute a VC crown decomposition $(I', J')$ in $G'$.
    If this algorithm finds a VC crown decomposition $(I', J')$, then we associate $(I', J')$ to the HS crown decomposition $(I, J)$, where $I$ consists of the vertices corresponding to the vertices in $I'$, and $J$ consists of the subedges corresponding to the vertices in $J'$.
    We now check the three conditions in the definition of an HS crown decomposition for $(I, J)$.
    \begin{enumerate}
        \item Since $A$ is an independent set in $H$ and $I\subseteq A$, we have that $I$ is an independent set.
        \item Since there is no edge between $I'$ and $V'\setminus (I'\cup J')$, let $J_{I} := \{Y \subseteq V : Y \cup \{x\} \in E \text{ for some }x \in I \}$. we have that $J_{I} = J$. 
        \item Since there is an injective mapping (matching) $M': J'\rightarrow I'$ such that, $\forall x\in J', \{x, M(x)\}\in E(G')$,
        there is an injective mapping (matching from subedges to vertices) $M: J\rightarrow I$ such that $\{M(X)\}\cup X \in E(H)$ holds for all $X\in J$ by transforming the vertices in $M'$ into the corresponding vertices or subedges in $H$.
    \end{enumerate}

    Since $I'\neq \emptyset$, we have that $I \neq \emptyset$.
    So $(I, J)$ is an HS crown decomposition of $V$ with $I \neq \emptyset$.

        If this algorithm cannot find such a VC crown decomposition, conclude that $(H, k)$ is a \textbf{no}-instance.
        Now, we show that this conclusion is correct. 
        By contradiction, suppose $(H, k)$ is a \textbf{yes}-instance.
        By Lemma \ref*{strict-key-bound}, there exists a strict HS crown decomposition in $H$.
        For any minimal strict HS crown decomposition $(I, J)$, by Lemma \ref*{LP-lemma}, we know that any optimal LP solution to SHS-LP sets all variables corresponding to vertices in $I$ to $0$, and all variables corresponding to vertices in $V(J)$ to $1$, which means that $I\subseteq A$ and $J\subseteq B$.
        We associate $(I, J)$ to the VC crown decomposition $(I', J')$ of $V'$,
        where $I'$ is the vertices corresponding to the vertices in $I$, $J'$ is the vertices corresponding to the subedges in $J$.
        It is easy to check that $(I', J')$ is a feasible VC crown decomposition in $G'$. 
        However, in this case, by Lemma \ref*{q-expansion-finding}, our algorithm will find a VC crown decomposition in $G'$, which contradicts the fact that our algorithm cannot find a VC crown decomposition.
\end{proof}

Thus, Lemma \ref{RRlp-lemma} directly implies the Reduction Rule \ref*{lp-reduction}.



\section{Conclusion}

In this paper, we present an improved vertex-kernel of size $(2d - 2)k^{d - 1} + k$  for \textsc{$d$-Hitting Set} for any fixed $d\geq 3$. 
Specifically, for \textsc{3-Hitting Set}, we obtain a vertex-kernel of size $4k^2 + k$.
Since many problems admit their best-known vertex-kernels via reductions from \textsc{$d$-Hitting Set},
our result also yields improved vertex-kernels for several related problems, such as \textsc{Triangle-Free Vertex Deletion}.

In our algorithm, we utilize LP to identify HS crown decompositions. LP serves as a powerful tool for deriving kernelizations. By leveraging the half-integrality property of LP solutions, a $2k$-vertex kernel for the \textsc{Vertex Cover} problem can be obtained~\cite{chen2001vertex}. Another example that exploits this property---specifically in the context of $k$-submodular relaxations---is the $2k^2 + k$-vertex kernel for the \textsc{Feedback Vertex Set} problem~\cite{DBLP:conf/icalp/Iwata17}. Furthermore, Kumar and Lokshtanov~\cite{DBLP:conf/iwpec/KumarL16} derived a $2dk$-vertex kernel for the \textsc{$d$-Component Order Connectivity} problem, for any fixed $d \geq 1$, by employing LP to compute a variant of crown decomposition for reduction. We believe that LP holds further potential for kernelization in a wider range of problems.

For future work, it remains an open question whether \textsc{$d$-Hitting Set} admits a kernel with $O(k^{d - 1 - \epsilon})$ vertices for some $\epsilon > 0$. Additionally, establishing nontrivial lower bounds on the size of vertex-kernels presents an interesting direction for further research.



\bibliography{dHS}

\end{document}